\documentclass[reprint,
 amsmath,amssymb,
 aps,
]{revtex4-2}
\usepackage{xcolor}

\usepackage{comment}
\usepackage{graphicx}
\usepackage{dcolumn}
\usepackage{bm}

\usepackage{enumerate} 
\usepackage{amsthm}
\newtheorem{theorem}{Theorem}
\newtheorem{preposition}{Proposition}

\begin{document}

\preprint{APS/123-QED}

\title{Master stability for traveling waves on networks}

\author{Stefan Ruschel}
\affiliation{
University of Leeds, United Kingdom
}%

\author{Andrus Giraldo}
\affiliation{%
Korea Institute for Advanced Study, South Korea
}%
\date{\today}

\begin{abstract}
We present a framework for determining effectively the spectrum and stability of traveling waves on networks with symmetries, such as rings and lattices, by computing master stability curves (MSCs). Unlike traditional methods, MSCs are independent of system size and can be readily used to assess wave destabilization and multi-stability in small and large networks.
\end{abstract}

\maketitle

Since the pioneering works of Fermi, Pasta, Ulam, and Tsingou \cite{fermi1955studies, zabusky1965interaction}, discrete traveling waves (DTWs) are studied in areas such as neuroscience \cite{erneux1993propagating,ermentrout2001traveling, avitabile2023bump}, nonlinear optics \cite{watanabe1996dynamics,alfaro2020pulse,parker2023standing}, chemical reactions \cite{totz2018spiral}, crystal dislocations \cite{aigner2003new,iooss2006normal,parker2021stationary}, Ising-like phase transitions \cite{bates1999discrete, Evans2015XY}, conservation laws \cite{sprenger2024hydrodynamics}, coupled oscillators \cite{laing2011fronts,perlikowski2010routes}, traffic jams \cite{orosz2005bifurcations}, etc. DTWs have received significant attention both analytically \cite{friesecke1994existence,chow1998traveling} and numerically \cite{abell2005computation}.
DTWs have also been observed in discretizations of continuous spatially extended systems, such as discrete-nonlinear-Schrödinger or Frenkel-Kantorova-type models \cite{ablowitz2004discrete,kevrekidis2009discrete,braun2004frenkel}.
Waves with time-varying profile (e.g. discrete
breathers) have also been considered \cite{flach1998discrete, Duran2022}.

Considering the importance of DTWs, there is a fundamental need for a unifying theory describing their spectral properties. Characterizing the stability of DTWs has remained challenging often tackled with specific approximations \cite{parker2023standing}. Rigorous results can be obtained for special cases; for instance, for zero wave speed \cite{kapitula2001stability} and in weakly coupled rings \cite{ermentrout1985behavior}. Similarly, the stability of DTWs has been studied previously on feed-forward ring networks \cite{klinshov2017embedding, perlikowski2010periodic}, and feed-forward discrete tori \cite{kantner2015delay}.

In this letter, we outline a master stability theory for DTWs on networks and coupled systems of ordinary differential equations that are invariant under cyclic index shift ($\mathbb{Z}_N$-equivariant) \cite{golubitsky2012symmetry}. Notable examples include systems of $N$ identical nodes coupled in a (possibly nonlocal) ring, spatially discretized PDEs with periodic boundary conditions, and systems with higher-order coupling topologies \cite{bick2023higher} invariant under cyclic index shift, as well as ---by little extension--- multi-chromatic networks, discrete tori, and systems with higher spatial embedding order (see Supplementary material for details). 

The proposed framework is a generalization of the master stability function approach for synchronization in networks (when all nodes behave dynamically the same) developed by Pecora and Carroll~\cite{pecora1998master}, and shares similarities with classic methods for assessing the stability of traveling waves in continuous systems \cite{whitham2011linear}. Specifically, we recast the higher dimensional Floquet spectral problem as a delay differential equation, possibly of mixed type, the master stability equation. Using the derived equation, we introduce a numerical scheme based on the continuation of a suitable two-point boundary value problem (2PBVP) to assess the stability of such waves independent of network size by computing master stability curves (MSCs).

\paragraph*{Setup.} To set the stage, we consider the $N$-dimensional network with $d$-dimensional internal dynamics given by 
\begin{align}
    x^\prime_n =& F(x_n) + \!\!\!\!\sum_{1 \leq |m| \leq r}\!\!\!\! H_m(x_n,x_{n+m}), \quad x_n(t)\in\mathbb {R}^d. \label{eq:x}  
\end{align}
Here, $x_n$, $n=1,2,\ldots, N$,  represents a node in the network, and $x_n^\prime$ denotes its first derivative with respect to time $t$. The function $F$ defines the uncoupled dynamics at each node $x_n$. Each node $x_n$ is coupled to nodes $x_{n+m}$ (using the convention $x_{j}=x_{j+N}$ for $j\in \mathbb{Z}$) through the functions $H_m$, $1 \leq |m| \leq r$, where $r$ is the coupling radius of the network. A comprehensive treatment of systems featuring higher-order interactions (see, e.g., \cite{bick2023higher}) can be found in the Supplementary material. Note that system~\eqref{eq:x} is equivariant under cyclic permutation $x_n\mapsto x_{n+1}$ (mod $N$); and hence supports DTWs \cite{golubitsky2012symmetry}. 

Our results pertain to the stability of DTWs. More precisely, let $\gamma=(x_1,\ldots,x_N)$ be a DTW of system~(\ref{eq:x}), where each component of $\gamma$ is of the form
\begin{align}
    x_n(t)=y(t-n\tau), \quad y(t)\in\mathbb {R}^d,\label{eq:y}  
\end{align}
where $y$ is the $T$-periodic profile of the wave propagating with constant wave speed $c=1/\tau$ (one unit of space per unit of time). In finite networks, the consistency relation $k=\tau/T$, $k\in \mathbb{Q}$, must hold, where $k$ is the wave number of the DTW ($1/k$ being the wavelength).


Figure~\ref{fig:1}(a)-(b) shows two examples of such DTWs with wave number $k=1/20$ and constant wave speed $c \approx 0.6764$ in system (\ref{eq:FHN}) with (a) $N=20$ and (b) $N=100$ nodes. In both cases, the wave profile is the same since,  for $k=M_0/N_0$ and  $\gcd(M_0,N_0)=1$, the DTWs corresponding to the profile $y$  can be supported in a network of size $N=KN_0$, where $K \in \mathbb{N}$, by concatenating $K$ copies of the profile. This does not change the wave number or period. As such, Fig.~\ref{fig:1}(a) shows the minimal network size realization $N_0=20$ of  profile $y$. 

Stability of the DTW $\gamma$ can be assessed by computing its Floquet exponent spectrum  (FES). Specifically, a Floquet exponent (FE) $\lambda \in \mathbb{C}$ and its nonzero bundle $z=(z_1, ... , z_N)$ can be computed as solutions of the $Nd$-dimensional Floquet 2PBVP 
\begin{equation}\label{eq:x-lin-unscaled} 
\begin{aligned}
    z^\prime_n(t) =& A(t-n\tau) z_n(t) + \!\!\!\!\sum_{1 \leq |m| \leq r}\!\!\!\! B_{m}(t-n\tau)z_{n+m}(t),\\
    z_n(t)=&e^{-\lambda T} z_n(t+T),
\end{aligned}
\end{equation}
where $z_n(t)\in\mathbb{C}^d$,
with $T$-periodic coefficient functions 
\begin{align}
A(t)=&DF(y(t))+ \!\!\!\!\sum_{1 \leq |m| \leq r}\!\!\!\! D_1 H_m(y(t),y(t-m\tau)),\nonumber\\
B_{m}(t)=&D_2 H_m(y(t),y(t-m\tau)).\nonumber
\end{align}
Here, $D_1 H_m$ and $D_2 H_m$ denote the Jacobians of $H_m$ with respect to its first and second arguments, respectively. The FES has size $Nd$ (counting multiplicity) and consists of all  FEs with imaginary parts inside the interval $[-\pi i/T, \pi i/T)$.  Particularly, if all Re$(\lambda)<0$ except for the trivial exponent $\lambda=0$ (Goldstone mode), then the DTW $\gamma$ is (orbitally) stable. There are multiple ways to numerically obtain the FES, ranging from explicitly solving the Floquet 2PBVP to finding the eigenvalues of the monodromy matrix \cite{KuznetsovBook2004}. In any case, the computations become infeasible or inaccurate for large system size $N$, or period $T$.

\paragraph*{Master stability equation.} Our approach overcomes these issues by leveraging Theorem 1 (Supplementary material), which states that every bundle $z$ of DTW $\gamma$ satisfying~(\ref{eq:x-lin-unscaled}) is one-to-one with a nonzero periodic solution of the problem 
\begin{equation}
\begin{aligned}  
    \zeta^\prime(t) =& \left(A(t)-\lambda I\right) \zeta(t) + \!\!\!\!\sum_{1 \leq |m| \leq r}\!\!\!\! e^{im\phi} B_{m}(t) \zeta(t-m\tau),\\
    \zeta(t)=&~ \zeta(t+T),
\end{aligned}\label{eq:v-master}
\end{equation}
the \emph{master stability equation (MSE),} where $\zeta(t)\in \mathbb{R}^d$ and $\phi=2\pi l/N,$ for some $l=1,...,N$. Equation~(\ref{eq:v-master}) serves as a generalization of the Master stability function by Pecora and Carrol to traveling waves on networks. Moreover, for a simple FE $\lambda$, its associated bundle $z$  is
\begin{align}
z_n(t)&= \zeta(t-n\tau)e^{\lambda t + i n\phi}\label{eq:planewave}
\end{align} 
for a periodic profile $\zeta$ satisfying the MSE~(\ref{eq:v-master}) for the corresponding value of $\lambda$ and $\phi$ (see Theorem~1 in Supplementary material). Equation (\ref{eq:planewave}) presents a projection on the network modes and a discrete analog of the plane wave ansatz in continuous systems \cite{whitham2011linear}. 

\begin{figure}[t]
\includegraphics[width=\linewidth]{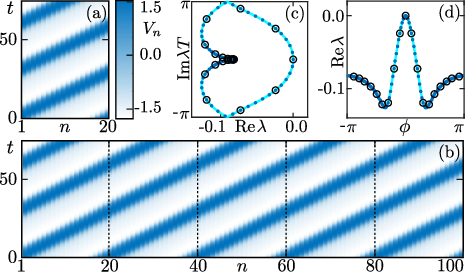}
\caption{\label{fig:1} Master stability curve (MSC) of DTWs in  system (\ref{eq:FHN}) of $N$ dissipatively-coupled FitzHugh-Nagumo oscillators $x_n=(V_n,W_n)$, $ V_n(t)$ and $W_n(t)\in \mathbb{R}$. (a)--(b): Space-time plots of DTWs with wave number $k=1/20$, (a): $N=20$ and (b): $N=100$. (c): Floquet exponents with Re$(\lambda)\geq-0.3$ of DTW (black circles) in (a), of DTW (blue dots) in (b), and MSC (cyan). (d):  FESs and dispersion relation, i.e. MSC in (c) shown in $(\phi, \mbox{Re}(\phi))$-projection.}
\end{figure}

The FE problem~(\ref{eq:x-lin-unscaled}) is effectively solved by considering the MSE (\ref{eq:v-master}) and treating $\phi$ as a continuous parameter.  This is exemplified in
Fig.~\ref{fig:1}(c), where we show the $N$ largest FEs (in real part) of the waves shown in (a) for $N=20$ (black circles) and (b) for $N=100$ (blue dots). Notice that both waves are stable and that FEs appear to delineate a curve suggesting an accumulation of FEs of DTWs as $N\to\infty$, a phenomenon previously observed in feed-forward rings of Duffing oscillators \cite{perlikowski2010routes}. Indeed, by continuously varying $\phi$, the solutions $(\phi,\lambda(\phi))$ of (\ref{eq:v-master}) trace a set of curves, the \emph{master stability curves} (MSCs), which contain the full FES for all network sizes $N=KN_0$, $K=1,2,\ldots$; see the MSC (cyan) in Fig.~\ref{fig:1}(c). This is an advantage of the MSC framework: Instead of computing the FESs of the DTWs with wave profile $y$ and wave number $k$ for networks with size $KN_0$, the MSCs contain all FESs for all sizes. Indeed, whenever $e^{i\phi}$ is a $KN_0$-th root of unity, $K=1,2,\ldots$, then the associated $\lambda(\phi)$ is a FE of the DTW in the network of size $KN_0$, and a solution of  (\ref{eq:x-lin-unscaled}).  Moreover, the spectrum in the limit $K\rightarrow \infty$ (approximating the lattice case $K=\infty$) densely fills the MSCs. Figure~\ref{fig:1}(d) shows the $(\phi,\mbox{Re}(\lambda(\phi)))$-projection of the MSC in panel~(c), where the equidistant spacing between FEs is easily observed. In analogy to continuous systems, we refer to this projection as (a branch of) the dispersion relation, relating a perturbation with wave number $\phi/2\pi$ to the corresponding growth rate.

\paragraph*{Computing MCSs by means of numerical continuation of a 2PBVP.}  
We use numerical methods for the continuation of periodic solutions to delay differential equations~\cite{Sieber2014}. Given the wave profile $y$ of a DTW $\gamma$, its period $T$, and wavenumber $k$, we require a solution triple $(\lambda,\zeta,\phi)$ of the MSE~(\ref{eq:v-master}) to initialize the 2PBVP. One such triple is $(\lambda,\zeta,\phi)=(0,y',0)$ where $y'$ can be computed using numerical differentiation. Fixing the norm and angle of $\zeta(0)$, pseudo-arclength continuation in $(\phi,\lambda)$ \cite{KrauskopfOsingaBook2007} defines a one-parameter family of solutions for (\ref{eq:v-master}), of which the projection onto the $(\mbox{Re}(\lambda),\mbox{Im}(\lambda))$-plane is an MSC. For more details on the numerical implementation see Supplementary material.

We remark that the FES of a DTW might be comprised of several, disconnected MSCs. In general, we find the number of MSCs to be of the order of $d$.  Additional MSCs can be initialized by computing triples $(\lambda,\zeta,\phi)$ for a given perturbation wave number $\phi$ (see Proposition 1, Supplementary material).

\paragraph*{Applications of the master stability framework.}  Below, we illustrate how the master stability framework allows studying the stability of wave profiles independently of a specific embedding into a given network of size $N$.  It is then immediately interesting to see how the MSC changes for different wave profiles. To that aim, we introduce the profile equation corresponding to Eq.~(\ref{eq:x}) in a coordinate frame  co-moving with wave profile  $y$, i.e.
\begin{equation}
    y^\prime(t) = F(y(t)) + \!\!\!\!\sum_{1 \leq |m| \leq r}\!\!\!\! H_m(y(t),y(t-m\tau)), \label{eq:profile}
\end{equation}
which can be readily obtained by imposing (\ref{eq:y}) on system (\ref{eq:x}). Notice that $\tau$ can be considered a continuous parameter in Eq.~(\ref{eq:profile}). This allows applying continuation techniques to Eq.~(\ref{eq:profile}) to study how profile and period change as a one-dimensional solution family when $\tau$ (or a system parameter) is varied. The consistency relation implies a variation of $k$ according to $k=\tau/T$. When $k=p/q\in\mathbb{Q}$, the corresponding profile can be embedded in a suitable network of size $N=Kq$, for any $K\in \mathbb{N}$, by concatenating $Kp$ wave profiles and sampling with time window $p\tau$. For simplicity, we present examples for $p=1$. Below, we show how simultaneous numerical continuation of (\ref{eq:v-master}) and (\ref{eq:profile})  in $\tau$ (and therefore the wave number) can be used to explore \emph{(A)} wave destabilization and \emph{(B)} the presence of stable bound multi-pulses.

\emph{(A) Wave destabilization.} 
Starting with the DTWs shown in Fig. \ref{fig:1}, we increase $k$ from $k_1=1/20$ to $k_3=1/15$, see Fig.~\ref{fig:2}(a),  and notice a qualitative change of the MSC of the underlying wave profile as shown in Fig. \ref{fig:2}(b). Approximately at $k_2\approx 0.0617$, the MSC displays a curvature change (from negative to positive) around $\lambda=0$, causing a small curve segment to extend into the positive half-plane. The wave profiles for which the MSC has positive curvature at $\lambda=0$ (to the right of $k_2$) are colored red in (a).  The change in curvature is well pronounced for $k=k_3$, see (b). Also plotted is the FES of the corresponding DTW in a ring of $N=15$ (black circles) and $N=60$ (purple dots), which are both unstable. In the network with fewer nodes, this causes the wave to die out after some transient which is largely determined by the distance of the initial condition used for simulation from the exact unstable wave profile, see (c). In the larger network, we observe another phenomenon: multi-stability of DTWs. Indeed, (d) shows the two-step transition of a DTW starting near an unstable profile with $k=k_3$, to $k=k_1$, and settling on $k=k_4$. Both wave profiles with wave numbers $k_1$ and $k_4$ are stable, but the simulation in (d) missed the basin of stability of the DTW at $k=k_1$. 

\begin{figure}[t]
\includegraphics[width=\linewidth]{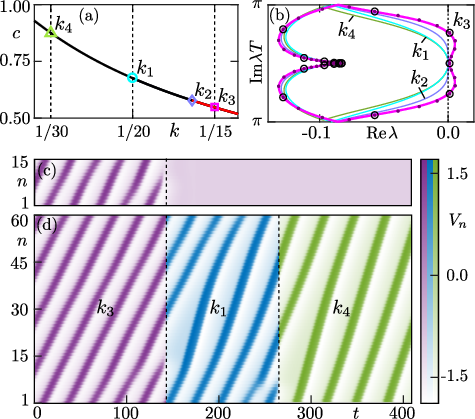}
\caption{\label{fig:2} Wave destabilization in system (\ref{eq:FHN}) of $N$ dissipatively-coupled FitzHugh-Nagumo oscillators $x_n=(V_n,W_n)$. (a): Numerical continuation of wave profile shown in Fig.~\ref{fig:1} parameterized by wave number $k$. The curve's color indicates the curvature of the MSC at $\lambda=0$; red (black) represents positive (negative) curvature. (b): MSCs of waves at values $k_1$ (cyan circle), $k_2$ (mauve diamond), $k_3$ (magenta square), and $k_4$ (green triangle) shown in (a). Also shown are the FEs of the wave profile with wave number $k=k_3$ for embedding dimension $N=15$ (black circles) and $N=60$ (purple dots). (c)--(d): Long-term dynamics of system (\ref{eq:FHN}) starting at unstable wave profile at $k_3$ for $N=15$ (c) and $N=60$ (d).}
\end{figure}

\begin{figure*}[t]
\includegraphics[width=\textwidth]{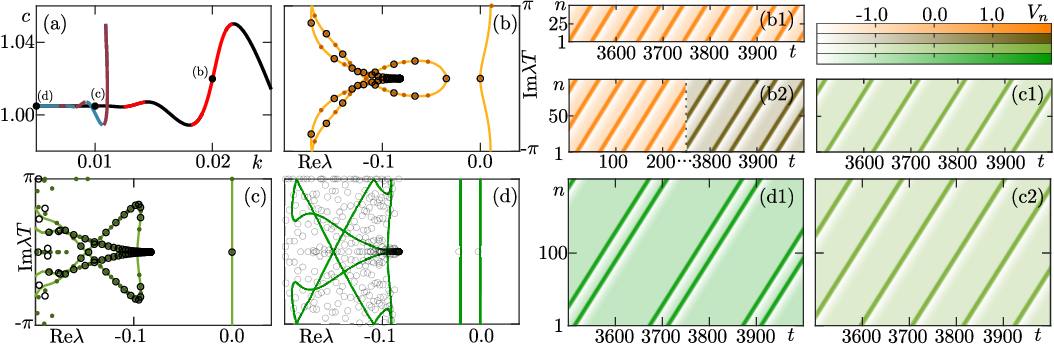}
\caption{\label{fig:3} Multi-stability and bound multi-pulses in system (\ref{eq:FHN}) of $N$ dissipatively-coupled FitzHugh-Nagumo oscillators $x_n=(V_n,W_n)$. (a): Numerical continuation of the wave profile shown in Fig.~\ref{fig:1} (main profile), and of the bound two-pulse wave profile (secondary profile). Colors represent the curvature of the MSCs attached to $\lambda=0$; namely, red (black) represents positive (negative) curvature for the main profile, while dark red (blue) represents positive (negative) curvature for the secondary profile. Positive curvature means that the curve extends to the right half-plane close to $\lambda=0$ indicating instability of the corresponding DTW for suffiently large embedding dimension $N$. (b): FES when $N=50$ (black circles) and $N=100$ (brown dots), together with MSCs (orange curves) of the main profile at $k=1/50$. (c): FES when $N=100$ (black circles) and $N=200$ (green dots), together with MSCs (light green curves) of the main profile at $k=1/100$. (d): FES (computed with same numerical tolerances as in (b) and (d)) when $N=200$ (gray circles) together with MSCs (dark green curve) of the secondary profile at $k=1/200$.   (b1)--(c1): Long-term dynamics of the main profile at points indicated in (a) for different embedding dimensions $N$. (d1): Long-term dynamics of the secondary profile as indicated in (a) in a network with $N=200$ nodes.}
\end{figure*}


\emph{(B) Wave localization and bound two-pulses.} Continuing the wave profile shown in Fig. \ref{fig:1} to smaller wave numbers $k$, we observe an oscillatory limiting behavior towards the asymptotic wave speed $c\approx1.0049$ as $k$ approaches $0$, see Fig.~\ref{fig:3}(a). Simultaneously, the minimal embedding dimension of the profile grows beyond bound, the wave profile localizes in a small number of nodes (a pulse) and its period tends to infinity. This oscillatory behavior can be attributed to the presence of a homoclinic bifurcation, presumably of Shilnikov-type, in the profile equation at $k=0$, possibly leading to multi-pulses near $k=0$ \cite{homburg2010homoclinic}. 
Another mechanism for the bifurcation of bound multi-pulses from a single pulse in delay differential equations has been studied in \cite{giraldo2023pulse}.
In what follows, we briefly illustrate how the presence of this object manifests in system~(\ref{eq:FHN}) for $k=0.02,0.01$ and $k=0.05$, and minimal embedding dimension $N=50,100$ and $200$ by showing the existence of stable bound two-pulses for sufficiently small $k$. The results are contained in Fig.~\ref{fig:3}.

Figure~\ref{fig:3}(b) shows part of the FES (with largest real parts) of wave profiles at $k=1/50$ alongside their corresponding MSCs. Notably, the MSC has split into two curves as compared to Fig.~\ref{fig:1}(c) (counting both segments on the left as one curve continuously parameterized by $\phi$). As part of this transition (details not shown), the curvature of the rightmost MSC has changed from negative to positive locally at $\lambda=0$, indicating that the corresponding DTWs are unstable for sufficiently large embedding dimension $N$. Indeed, the right-most curve satisfies the following property: $\lambda=0$ whenever $e^{i \phi}$ is a $50$th-root of unity.  This implies that the DTW is stable in the network with $N=50$ nodes. However, embedding the same wave profile in a network with $N=100,$ we observe that the corresponding DTW is unstable with a FE $\lambda=\lambda_0$ of the form Re$(\lambda_0)>0,$ Im$(\lambda_0 T)=i\pi$ whenever $\phi$ and is an odd $100$-th root of unity. This is evidenced by the FES (black circles and orange dots) in (b) and direct simulations starting near the wave profile (b1)--(b2). 

In (b2), we observe a DTW with non-equidistant pulses (a bound two-pulse). Interestingly, this DTW cannot correspond to the wave profile at $k=1/100$ in (a), which is orbitally stable when $N=100$ as evidenced by the FES (c) and direct simulation (c1); rather, it can be attributed to a period-doubling bifurcation. The two-pulse solution obtained in (b2) can be used to reconstruct a wave profile corresponding to the secondary curve of bound two-pulse wave profiles in (a) (dark red and blue). Furthermore, we observe that the bound two-pulse wave profile approaches a limit as $k$ tends to zero; however, both profiles look different as they approach it. To illustrate this, we compare direct simulations in a $N=200$ network when simulating near the original wave profile with $k=1/100$ (c2) and the bound two-pulse wave profile with $k=1/200$ (d1). For the bound two-pulse wave profile, the DTW is stable in its minimal network embedding, evidenced by the FES shown in (d) and simulation. On the other hand, the original wave profile is unstable in the $N=200$ network, where its largest unstable FE is close to zero in the real part. Indeed, the MSC in (c) attached to $\lambda=0$ has become flat; this suggests a scaling of the real parts of FEs $\lambda$, namely Re$(\lambda)=\mathcal{O}(k)$ on the curve, as $k=1/(cT)\to 0$ for finite values of $c$.  This can result in long transients where the solution stays close to the profile, as shown in (c2). Panels (c) and (d) again emphasize the importance of our approach where classical methods for computing the FES begin to fail, yet all FEs must, in truth, lie on an MSC computed with our framework as a consequence of Theorem 1 (see Supplementary material).

\paragraph*{Concluding Remarks.}
The spectrum of traveling waves on finite networks (DTWs) with index-shift invariance is determined by the master stability equation and master stability curves (MSCs). These curves fully encapsulate the spectrum of DTWs and can be computed effectively using numerical path-following techniques for delay differential equations. Specifically, the original high-dimensional Floquet exponent problem is replaced by a low-dimensional boundary value problem with periodic boundary conditions and additional free parameters, in particular $\phi$. This significantly reduces the computation time of analyzing the stability of DTWs  and can also be applied in other contexts, see Supplementary material. While DTWs on a ring serve as primary examples, the method is readily applicable to systems with higher embedding dimension, such as discrete tori and beyond, as well as higher-order (beyond pairwise) interactions. 
Several extensions of this work (e.g., scaling of FEs, general equivariant networks, and time-delayed couplings) warrant future research, in particular, the relation of the MSC with a spectrum of DTWs on infinite lattices.

\paragraph*{Example.} Figures in this letter are computed for $N$ diffusively-coupled FitzHugh-Nagumo (FHN) oscillators 
 $x_n=(V_n,W_n)$, $1\leq n\leq N$ (mod $N$) in the excitable regime, governed by
\begin{equation}
\begin{aligned}
V_n'=& V_n-\frac{V_n^3}{3} -W_n + 0.2 + V_{n+1}-2V_n + V_{n-1}\\
W_n'=& 0.08(V_n+0.7-0.8 W_n),
\end{aligned} 
\label{eq:FHN}
\end{equation}
where $V_n(t),W_n(t)\in\mathbb{R}$. The FHN system has been used to describe phenomenologically the propagation of action potentials along nerve axons \cite{fitzhugh1961impulses,nagumo1962active, coombes2023neurodynamics,hupkes2013stability}.

\begin{acknowledgments}
\vspace{-0.1cm}
S. R. and A. G. contributed equally to the paper as main authors. S.R. was supported by UKRI Grant No. EP/Y027531/1. A.G. was supported by KIAS Individual Grant No. CG086102 at Korea Institute for Advanced Study. The authors thank Matthias Wolfrum and two anonymous referees for valuable comments on the preprint. The code and the data that support the findings of this study are openly available at the following URL/DOI:
\url{https://github.com/andrusgiraldo/RG_MasterStability}.
\end{acknowledgments}

%

\appendix

\onecolumngrid

\vspace{1cm}
{\Large \textbf{Supplementary material: Master stability for traveling waves on networks}}

\section{Master stability equation for traveling waves in $\mathbb{Z}_N$-equivariant networks: Derivation, Main Theorem and proof}\label{sec:derivation}

We denote the state of a system $x=(x_1,...,x_N)$, where each $x_n \in \mathbb{R}^d$ may correspond to the state of a node $n$ in a network of size $N\geq 2$.  
For a coupling radius $r \in\mathbb{N}$, consider a sufficiently regular vector-valued function 
\begin{equation*}
    \begin{aligned} 
       G:\prod^{2r+1}_{j=1} \mathbb{R}^d &\rightarrow \mathbb{R}^d,
       (a_1,...,a_{2r+1}) \mapsto G(a_1,...,a_{2r+1}),
    \end{aligned}     
\end{equation*}
 and the vector field $g$ where each $x_n'$ is given by
\begin{align}
    x^\prime_n =& 
 g_n(x_1,...,x_N):= G(x_{n-r},...,x_{n-1},x_{n},x_{n+1},...,x_{n+r}), \tag{S1}\label{eqApp:ProblemGeneral}
\end{align}
where indices outside the range $1,\ldots,N$ are understood as modulo $N$ throughout the text, in particular, $x_0=x_N, x_{-1}=x_{N-1}, x_{-2}=x_{N-2}$, etc. Notice that the functions $g_n$ position the $n$-component of $x$ to be in the middle of the inputs of $G$.

System~(\ref{eqApp:ProblemGeneral}) includes a large class of possible configurations; for example, system~(1) (presented in the main text) is readily obtained by defining $G$ as $$G(a_1,a_2,...,a_{2r+1})=F(a_{r+1}) + \!\sum_{1 \leq m \leq r}\! (H_m(a_{r+1},a_{r+1+m}) +  H_{-m}(a_{r+1},a_{r+1-m})),$$where $F$ and $H_m$, $1\leq |m| \leq r$ are sufficiently regular vector-valued functions. More generally, other types of networks that might not be of the form of (\ref{eqApp:ProblemGeneral}) at first glance can be brought into such form.  For example, multi-chromatic networks with $L$-types of nodes (periodically repeating) and $NL$-many nodes in total can be put in this form by considering collections of $L$-nodes as a single unit. Moreover, for two-dimensional (or higher-dimensional) discrete tori with $L_1$ rows and $L_2$ columns, the same method can be applied by grouping columns or rows. For the special case when $L_1$ and $L_2$ are coprime, one can always find a permutation of the indices that allows us to recast it in form (\ref{eqApp:ProblemGeneral}). 

Let $\sigma$ be the permutation mapping that acts as a right shift on the state vector. That is, $\sigma(x) = (x_2,x_1)$ when $N=2$, and $\sigma(x) = (x_2,...,x_N,x_1)$ when $N>2$. By construction, system~(\ref{eqApp:ProblemGeneral}) is equivariant under the group generated by the mapping $\sigma$ ($\mathbb{Z}_N$-equivariant). Vector fields of this form allow for the existence of  discrete traveling waves (DTWs)\cite{golubitsky2012symmetry}; that is, they admit the existence of a periodic solution $\gamma$ such that $$\gamma(t)=(x_1(t),x_2(t),\ldots,x_N(t))=(y(t-\tau),y(t-2\tau),\ldots,y(t)),$$ for some periodic function $y$, called the wave profile, with minimal period $T$ and wave speed $1/\tau$. Since $y(t)=x_1(t+\tau)=x_2(t+2\tau)=...=x_N(t+N\tau)=x_N(t)$, $T$ and $\tau$ must be in a rational relation; that is, $\tau = kT$ for some $k\in \mathbb{Q}$. The number $k$ is commonly referred to as the wave number of the DTW $\gamma.$ 

It is straightforward to show that the profile $y$ satisfies  the differential equation
\begin{align}
    y^\prime(t) =& G(y(t+r\tau),\ldots,y(t+\tau),y(t),y(t-\tau), \ldots,y(t-r\tau)). \label{eqApp:profileGeneral} \tag{S2} 
\end{align}

Associated with the DTW $\gamma$, there exists a collection of numbers called the Floquet multipliers (FMs). A FM $\mu$ is a complex number with an associated nonzero solution $z=(z_1,z_2,...,z_N)$---its Floquet bundle--- of the variational problem
\begin{equation}\label{eqApp:VarProblem}
\begin{aligned}
    z^\prime_n(t) &= \sum_{0 \leq |m| \leq r} A_m(t-n\tau)z_{n+m}(t)\\
    z(0)&=\mu^{-1} z(T),
\end{aligned} \tag{S3} 
\end{equation}
where
\begin{align}
A_{m}(t)=& D_{r+m+1} G(y(t+r\tau),\ldots,y(t),\ldots,y(t-r\tau)).\nonumber
\end{align}
Here, $D_{q}$ represents the partial derivatives matrix of $G$ with respect to its $q$-th $d$-dimensional input vector $x_n$. By Floquet Theory, a DTW in system~(\ref{eqApp:ProblemGeneral}) has $Nd$ FMs (up to multiplicity) \cite{KuznetsovBook2004}. If all the FMs are strictly smaller than one in modulus, except for the trivial multiplier (which is always one), then $\gamma$ is orbitally stable.

Alternatively to the Floquet multiplier spectrum, it is often convenient to consider the Floquet exponent spectrum of values $\lambda=\mbox{Log}(\mu)/T$, where $\mbox{Log}$ is the fundamental branch of the complex logarithm. Notice that FEs $\lambda$ are uniquely defined up to addition by a multiple of $2\pi i$, and hence, it is convenient to restrict the following considerations to FE in a suitable strip of the complex plane, e.g. $\mbox{Im}(\lambda T)\in[-\pi,\pi).$ 

In the following, we are going to derive transformations that will allow us to prove our main theorem, Theorem~\ref{thm:MainSupp}. One can restate the Floquet multiplier problem~(\ref{eqApp:VarProblem}) into a periodic boundary value problem (PBVP) by rescaling system~(\ref{eqApp:VarProblem}) using the Floquet exponent  $\lambda=\mbox{Log}(\mu)/T$ of $\mu$.  More specifically, considering the rescaled variable $u(t)=R_0(z(t)):= e^{-\lambda t}z(t)$, then $u=(u_1,u_2,...,u_N)$ is a nonzero periodic solution of the system of equations 
\begin{equation}\label{eqApp:VarProblemScaled}  
    \begin{aligned} 
        u^\prime_n(t) &= -\lambda \mathbb{I}_d u_n(t) + \sum_{0 \leq |m| \leq r} A_m(t-n\tau)u_{n+m}(t),\\ 
        u(0) &= u(T).
    \end{aligned}\tag{S4}       
\end{equation}
Notice that the problem of the stability of DTWs reduces to find values of $\lambda \in \mathbb{C}$ where system~(\ref{eqApp:VarProblemScaled}) has nontrivial periodic solutions. 
Due to the relation between FMs and exponents, the stability of $\gamma$  can be read from the sign of the real part of exponents; that is, if the sign of the real part of all the Floquet exponents is negative, except for the trivial exponent (which is zero), then $\gamma$ is orbitally stable.

Primarily, we consider system~(\ref{eqApp:VarProblemScaled}) and introduce the component-wise time-shift transformation $R_1$, that sends the bundle $u$ to a time-shifted bundle $v=(v_1,\dots,v_N)$, i.e., $v(t)=R_1(u(t))$, that satisfies 
$$v_n(t):=u_n(t+n\tau), \quad \forall n\in \{1,\ldots,N\}$$
Particularly, computing the derivative of $v_n$ gives

\begin{equation*}
    \begin{aligned} 
        v^\prime_n(t) &= u'_n(t+n \tau) \\
          &= -\lambda \mathbb{I}_d v_n(t)+ \sum_{0 \leq |m| \leq r} A_m(t)u_{j+m}(t+n\tau). 
    \end{aligned}      
\end{equation*}
Since $v_{n+m}(t-m\tau)=u_{n+m}(t+n\tau)$, where each component of $v=(v_1, ..., v_N)$ satisfies the periodic boundary problem
\begin{equation}\label{eqApp:VarProblemScaled3}  
    \begin{aligned} 
        v^\prime_n(t) &= -\lambda \mathbb{I}_d v_n+ \sum_{0 \leq |m| \leq r} A_m(t)v_{n+m}(t-m\tau), \\
        v_n(0) &= v_n(T).
    \end{aligned}\tag{S5}      
\end{equation}
The transformation $R_1$ is a bijection between solutions of (\ref{eqApp:VarProblemScaled}) and (\ref{eqApp:VarProblemScaled3}). Notice that the matrices $A_m(t)$ do not depend on index $n$; and hence, we can rewrite system~(\ref{eqApp:VarProblemScaled3}) as
\begin{equation}\label{eqApp:VarProblemScaled4}  \tag{S6}
\begin{aligned} 
        v^\prime(t) &= -\lambda\mathbb{I}_{dN}v(t)
    + \sum_{0 \leq |m| \leq r} \left[C^{-m} \otimes A_m(t)\right] v(t-m\tau)\\
    v(0) &= v(T)
    \end{aligned}
\end{equation}
where $\otimes$ is the Kronecker product of matrices, and $C$ is the $N$-dimensional circulant (companion) matrix
\begin{equation*}
C=\left(\begin{matrix}
    0 & 0 & \dots  & 0 & 1\\
    1 & 0 & \ldots & 0 & 0\\
    0 & 1 & \ldots  & 0 & 0\\
    \vdots & \vdots & \ddots  & \vdots & \vdots \\
    0 & 0 & \ldots  & 0 & 0\\
    0 &  0 & \ldots  & 1 & 0
\end{matrix}\right).
\end{equation*}

Given that $C$ is diagonalizable, then system~(\ref{eqApp:VarProblemScaled4}) can be block diagonalized. Hence, there is a linear bijection $R_2$ that sends the time-shifted bundle $v$ to a new bundle $\hat{v}(t)=R_2(v(t))$ that satisfies 
\begin{equation}\label{eqApp:VarProblemScaled5} \tag{S7} 
    \begin{aligned} 
        \hat{v}^\prime(t) &= -\lambda\mathbb{I}_{dN} \hat{v}(t)+ \sum_{0 \leq |m| \leq r} \left[D^{-m} \otimes A_m(t)\right] \hat{v}(t-m\tau)\\  
    \hat{v}(0) & = \hat{v}(T)
    \end{aligned} 
\end{equation}
where $D$ is a $N$-dimensional diagonal matrix with entries $e^{-2\pi i/N},e^{-2\pi i (2/N)},..., e^{-2\pi i (N-1)/N}$ and $e^{-2\pi i N/N}=1$, i.e., the eigenvalues of $C$.  More precisely, $R_2$ takes the form of the discrete  Fourier transform, where the components of $\hat{v}=(\hat{v}_1,\ldots, \hat{v}_N)$ are given by 
\begin{equation*}
    \hat v_l=\sum_{n=1}^N v_n e^{-2in \pi l /N},\quad \hat v_l(t) \in \mathbb{R}^d 
\end{equation*}
with inverse transformation $R_2^{-1}$ acting on each component of $v$ as
$$v_n=\frac{1}{N}\sum_{l=1}^N \hat v_l e^{2in \pi l /N},\quad v_n(t) \in \mathbb{R}^d$$
where $l=1,\ldots,N$. System~(\ref{eqApp:VarProblemScaled5}) is block diagonal and consists of $N$ decoupled (and thus independent) equations for the $d$-dimensional states $\hat v_l (t),$ where

\begin{equation}\label{eqApp:VarProblemScaled6} \tag{S8}
    \begin{aligned} 
        \hat{v}^\prime_l(t) &= -\lambda\mathbb{I}_d \hat{v}_l(t)
    + \sum_{0 \leq |m| \leq r} (e^{-2\pi i l/N})^{-m} A_m(t)\hat{v}_{l}(t-m\tau), \\
     \hat{v}_l(0) &= \hat{v}_l(T).
    \end{aligned}      
\end{equation}

By using  Eq.~(\ref{eqApp:VarProblemScaled6}) as a reference, we define the master stability equation (MSE) for traveling waves as 
\begin{equation}\label{eqApp:masterstability} \tag{S9}
    \begin{aligned} 
        \zeta^\prime(t) &= -\lambda\mathbb{I}_d \zeta(t) + \sum_{0 \leq |m| \leq r} e^{im\phi} A_m(t)\zeta(t-m\tau),\\
        \zeta(0) &= \zeta(T).
    \end{aligned}
\end{equation}
Note that every solution triple $(\lambda, \phi,\zeta)$ of MSE~(\ref{eqApp:masterstability}) with nonzero periodic solution $\zeta$ and $\phi=2\pi l /N$, for some $l=1,...,N$,  directly corresponds to a solution triple $(\lambda,l,v_l)$ of (\ref{eqApp:VarProblemScaled6}) by setting $l=\phi/(2\pi)$ and $v_l=\zeta$. Now that we have all the necessary ingredients, we can prove the following theorem.

\begin{theorem}\label{thm:MainSupp}
For a network of size $N$, let $\gamma$ be a DTW solution of system~(\ref{eqApp:ProblemGeneral})  with wave profile $y$, wave speed $1/\tau$, and period $T$. Then the following holds:
\begin{enumerate}[(i)]
\item For any triple $(\lambda,\phi,\zeta)=(\lambda^*,\phi^*,\zeta^*)$ that satisfies (\ref{eqApp:masterstability}) with nonzero periodic solution $\zeta^*$ and $\phi^*=2\pi l/N$, for some $l=1,2,...,N$, there exists a corresponding tuple $(\mu,z)=(e^{\lambda^* T},z^*)$ that satisfies the Floquet problem~(\ref{eqApp:VarProblem}) and each component of $z^*=(z^*_1,...,z^*_N)$  is given by
\begin{align*}
z^\ast_n(t)&= \zeta^*(t-n\tau)e^{\lambda^* t+i n\phi^*}, \quad\quad n=1,\dots,N.
\end{align*}
\item For any tuple $(\mu,z)=(\mu^*,z^*)$, $z^*=(z^*_1,...,z^*_N)$, that satisfies the Floquet problem~(\ref{eqApp:VarProblem}),  there exists a triple $(\lambda,\phi,\zeta)=(\mbox{Log}(\mu^*)/T,2\pi l^* / N,\zeta^*)$, for some $l^*=1,...,N$ and nonzero periodic solution $\zeta^*$, that satisfies  system~(\ref{eqApp:masterstability}).
\item More strongly, any tuple $(\mu,z)=(\mu^*,z^*)$, $z^*=(z^*_1,...,z^*_N)$, where $\mu^*$ is a simple Floquet multiplier, that satisfies the Floquet problem~(\ref{eqApp:VarProblem}) must be of the form:
\begin{align*}
z^\ast_n(t)&= \zeta^*(t-n\tau)e^{\lambda^* t+i n\phi^*}, \quad\quad n=1,\dots,N. 
\end{align*}
with $\lambda^*=\mbox{Log}(\mu^*)/T$, $\phi^*=2\pi l^*/N$ for a unique $l^*=1,2,...,N$, and the triple $(\lambda,\phi,\zeta)=(\lambda^*,\phi^*,\zeta^*)$ satisfying~Eq.~(\ref{eqApp:masterstability}). Moreover, 
\begin{equation*}
l^* =\frac{iN}{2\pi }\left(\frac{\tau}{T} {\rm{Log}} \left( \mu^* \right)-{\rm{Log}}\left( \frac{\langle z^*_1(\tau),z^*_N(0) \rangle}{||z^*_N(0)||^2}\right)\right).
\end{equation*}
\end{enumerate}
\end{theorem}
\begin{proof}
\begin{enumerate}[(i)]
\item Let $(\lambda,\phi,\zeta)=(\lambda^*,\phi^*,\zeta^*)$ be a triple satisfying system~(\ref{eqApp:masterstability}) with nonzero $\zeta^*$ and $\phi^*=2\pi l^* /N$, for some $l^*=1,...,N$.  Then we can construct a periodic solution $\hat{v}^*=(\hat{v}^*_1,\ldots,\hat{v}^*_N)$ for system~(\ref{eqApp:VarProblemScaled5}) as follows
$$\hat{v}^*_{l^*}(t)=N \zeta^*(t), \text{ and }\hat{v}^*_l(t)=0 \text{ for $l \neq l^*$ }.$$
Since $\hat{v}^*$ is nonzero periodic solution of system~(\ref{eqApp:VarProblemScaled5}) by construction, and $R_0$, $R_1$ and $R_2$ are linear bijections, the function 
$$z^*(t)=R_0^{-1}R_1^{-1}R_2^{-1}\hat{v}^*(t)=(\zeta^\ast(t-\tau)e^{\lambda^* t+i \phi},\zeta^\ast(t-2\tau)e^{\lambda^* t+2i\phi},\ldots,\zeta^\ast(t-(N-1)\tau)e^{\lambda^* t+i (N-1)\phi},\zeta^\ast(t)e^{\lambda^* t}),$$ 
is a nonzero periodic solution of system~(\ref{eqApp:VarProblemScaled}) for $\mu=e^{\lambda^* T}$. Hence, $(\mu^*,z^*)$ is a Floquet tuple that satisfies the Floquet problem~(\ref{eqApp:VarProblem}), and the assertion follows.
\item Let $(\mu^*,z^*)$ be a tuple that satisfies the Floquet problem~(\ref{eqApp:VarProblem}) with nonzero $z^*$,  then $\hat{v}^*(t)=R_2R_1R_0 z^*(t)$ is a nonzero periodic solution with minimal period $T$ that satisfies system~(\ref{eqApp:VarProblemScaled6}) for $\lambda=\mbox{Log}(\mu)/T$. Hence, there exist at least one index $l_1$ such that a component $\hat{v}^\ast_{l_1}$ of $\hat{v}$ is nonzero.  Then the triple $(\lambda,\phi,\zeta)=(\mbox{Log}(\mu^*)/T,2\pi l_1 / N,\hat{v}^\ast_{l_1})$  satisfies system~(\ref{eqApp:masterstability}) and the assertion follows.
\item To prove (iii), we continue from the end of (ii); that is, from a tuple $(\mu,z)=(\mu^*,z^*)$ we construct a triple $(\lambda,\phi,\zeta)=(\mbox{Log}(\mu^*)/T,2\pi l_1 / N,\hat{v}^\ast_{l_1})$ that satisfies system~(\ref{eqApp:masterstability}). By (i), there exists a bundle $\tilde{z}=(\tilde{z}_1,...,\tilde{z}_N)$ such that
\begin{align*}
\tilde{z}_n(t)&= \zeta(t-n\tau)e^{\lambda t+in \phi}=\hat{v}^\ast_{l_1}(t-n\tau)e^{t {\scriptstyle \rm{Log}}(\mu)/T +in 2\pi l_1/ N}, \quad\quad n=1,\dots,N.
\end{align*}
and $(\mu^*,\tilde{z})$ satisfies (\ref{eqApp:VarProblem}). 
Simplicity of $\mu^*$ implies uniqueness of the associated Floquet eigenvector $z^*(0)$ up to scalar multiple $k \in \mathbb{C}\setminus\{0\}$ and, thus, $z^*(0)=c \tilde{z}(0)$. By using the uniqueness of solutions and linearity of~(\ref{eqApp:VarProblem}), we have $z^*= c \tilde{z}$.

Let $\tilde{\zeta}(t) = c \hat{v}^\ast_{l_1}(t)$, then the components of $z^*=(z^*_1,...,z^*_N)$ are of the form
\begin{align*}
z^*_n(t)&= \tilde{\zeta}(t-n\tau)e^{t{\scriptstyle\rm{Log}}(\mu)/T +i  n 2\pi l_1/ N}, \quad\quad n=1,\dots,N.
\end{align*}
By linearity of (\ref{eqApp:masterstability}), if $(\lambda,\phi,\zeta)=(\mbox{Log}(\mu^*)/T,2\pi l_1 / N,\hat{v}^\ast_{l_1})$ satisfies (\ref{eqApp:masterstability}) then $(\lambda,\phi,\zeta)=(\mbox{Log}(\mu^*)/T,2\pi l_1 / N,c\hat{v}^\ast_{l_1})=(\mbox{Log}(\mu^*)/T,2\pi l_1 / N,\tilde{\zeta})$ is also a triple satisfying (\ref{eqApp:masterstability}). Given that $z^*$ can be written as above, then $R_2R_1R_0 z(t)$ satisfies that all blocks of (\ref{eqApp:VarProblemScaled6}) are zero except for one (the component $l_1$); thus, every $(\mu^*,z^*)$ is associated to a unique $l_1$. Moreover, $l_1$ can be computed as 
$$l_1=\frac{iN}{2\pi }\left(\frac{\tau}{T}\text{Log}\left( \mu^* \right)-\text{Log}\left( \frac{\langle z^*_1(\tau),z^*_N(0) \rangle}{||z^*_N(0)||^2}\right)\right),$$ 
from the tuple $(\mu,z^*)$. Hence, the assertion follows.
\end{enumerate}
\end{proof}

Intuitively, Theorem~\ref{thm:MainSupp} can be understood as follows:
\begin{itemize}
\item Part (i) states that by finding solutions triples $(\lambda, \phi,\zeta)$ of the MSE~(\ref{eqApp:masterstability}), where $\zeta$ is a non-zero periodic solution and $\phi$ is of the form $2\pi l/N$; we have effectively found a Floquet exponent $\lambda$ of the traveling wave $\gamma$.
\item Part (ii) states that any Floquet multiplier $\mu$ with bundle $z$ is mapped through its associated Floquet Exponent $\lambda$ to a solution triple of (\ref{eqApp:masterstability}). That is, the MSE~(\ref{eqApp:masterstability}) contains the information of all the Floquet exponents of the original problem.
\item Part (iii) shows an exact correspondence between the solutions of the MSE and the Floquet bundle of the original problem. Furthermore, we can, therefore, think of each simple Floquet exponent $\lambda$ as being directly associated with exactly one $\phi$, with $\text{Im}(\phi)\in [0,2\pi)$.
\end{itemize}

As outlined in the main manuscript, Theorem~(\ref{thm:MainSupp}) has direct implications for the spectrum of all DTWs with profile $y$ that can be embedded in networks of different sizes that satisfy Eq.~(\ref{eqApp:ProblemGeneral}). More precisely, let $y$ be wave profile with period $T$, wave speed of $1/\tau$, and wave number $k=M_0/N_0$, where $M_0,N_0 \in \mathbb{N}$ and $\gcd(M_0, N_0)=1$, that satisfies Eq.~(\ref{eqApp:profileGeneral}).  In this case, one can always construct a DTW $\gamma_{KN_0}$ for a network with size $KN_0$, $K\in\mathbb{N}$, by imposing 
$$\gamma_{KN_0}(t)=(x_1(t),x_2(t),\ldots,x_{KN_0}(t))=(y(t-\tau),y(t-2\tau),
\ldots,y(t-(KN_0-1)\tau),y(t)),$$ 
which effectively corresponds to concatenating $K$ instances of the periodic profile $y$. Indeed, $\gamma_{KN_0}$ is a DTW of system~(\ref{eqApp:profileGeneral}) for size $KN_0$, and neither $y$, $\tau$, $T$ or $k$ has changed. Furthermore, by construction, all $\gamma_{KN_0}$ give rise to the same MSE~(\ref{eqApp:masterstability}); the only distinction lies in the corresponding phases $\phi$ that must be considered to compute the respective Floquet exponents. Thus, the MSE contains all the Floquet spectrum information of the DTWs $\gamma_{KN_0}$  for all network sizes $K N_0$, allowing us to view the MSE as an intrinsic property associated with $y$ and its characteristics, rather than being dependent on the network size.

Finally, we also bring to attention the fact that even though the expression of the MSE~(\ref{eqApp:masterstability})  resembles the Floquet exponent problem of a delay differential equation \cite{sieber2011characteristic} at first glance, they are not equivalent. Note that Eq.~(\ref{eqApp:masterstability}) does not contain terms $e^{-\lambda\tau}$ and as such does not correspond to the periodic boundary value FP of a delay differential equation. As a result, the number of linearly independent solutions can be finite, namely, $Nd$ (counting multiplicity) in this case. 
The statement follows directly from Theorem 1 and the fact that the delay parameter $\tau$ in the master stability equation is in a rational relationship with the period, i.e. $N\tau = M T,$ where $k=M/N$. The expression $$z_n(t)=\zeta(t-n\tau)e^{\lambda t -in\phi}$$ in Theorem 1 (i), then allows to embed any solution $\zeta$ of (\ref{eqApp:masterstability}) into a linear $Nd$-dimensional system of ordinary differential equations and, thus, the assertion follows. This is in sharp contrast to the variational problem for delay differential equations where terms $e^{\lambda\tau}$ lead to infinitely many solutions \cite{sieber2011characteristic}.

\section{Numerical Implementation and Continuation of Master Stability Curves}\label{app:2PBVP}
A classical method to compute the Floquet (multiplier) spectrum is to compute the eigenvalue spectrum of the so-called monodromy matrix. Let $\Xi(t)$ be the fundamental matrix solution of the nonautonomous linear ODE (\ref{eqApp:VarProblem}), and denote $Z(t,z(0))=\Xi(t)z(0)$ the solution of (\ref{eqApp:VarProblem}) with initial condition $z(0)$. The matrix $M=\Xi(T)$ ($T$ being the temporal period of $\gamma$) is called the \emph{monodromy matrix}. The Floquet boundary conditions imply that $Mz(0)=z(T)=\mu z(0),$ i.e., FMs of (\ref{eqApp:VarProblem}) correspond to the eigenvalues of $M$. The columns of $M$ can be constructed as the image of the canonical basis eigenvectors $e_j\in\mathbb{R}^{Nd},$ $j=1,\ldots,Nd$ where $(e_j)_l=\delta_{jl}$, that is $M=[Z(T,e_1)|\ldots|Z(T,e_{Nd})]$. As a result, computing the monodromy matrix (and its eigenvalues) is typically expensive (and it could be inaccurate) when $N$ is large. 

\subsection{Two-point boundary value problem for Master stability curves}
The master stability curves (MSCs), showcased in the main manuscript, circumvent the network size issue for computing the Floquet exponent spectrum by leveraging the MSE~(\ref{eqApp:masterstability}) and Theorem~\ref{thm:MainSupp}. As discussed in the main manuscript, the MSCs are one-dimensional curves in the $(\mbox{Re}(\lambda),\mbox{Im}(\lambda))$-plane, which contain the  Floquet exponent spectrum of a DTW. To compute the MSCs, numerical continuation of a suitable two-point boundary value problem \cite{KrauskopfOsingaBook2007} is employed by using $\phi$ as a parameter. More precisely, the two-point boundary value problem (2PBVP) takes the form
\begin{align}
    y^\prime(t) =& TG\left(y\left(t+r\frac{\tau}{T}\right),\ldots,y\left(t+\frac{\tau}{T}\right),y(t),y\left(t-\frac{\tau}{T}\right),\ldots,y\left(t-r\frac{\tau}{T}\right)\right), \tag{S10} \label{eq:BVP1}\\ \zeta^\prime(t) =& T\!\left(\!-\lambda\mathbb{I}_d \zeta(t) + \!\sum_{m=0}^{N-1} e^{-im\phi}  A_m(t)\zeta\left(t-m\frac{\tau}{T}\right)\!\right)\!, \tag{S11}\label{eq:BVP2} \\
    y(0) =& y(1), \tag{S12}\label{eq:BVP3} \\
    \zeta(0) =& \zeta(1), \tag{S13}\label{eq:BVP4}\\
    0 =& \int^1_0 \langle y, y^\prime_{ref}\rangle + \langle \zeta, \zeta^\prime_{ref} \rangle dt, \tag{S14}\label{eq:BVP5}\\
    1 =& T ||\zeta||^2_{L_2}, \tag{S15}\label{eq:BVP6}\\
    0 =& \text{Im}(\pi_1(\zeta(0))). \tag{S16}\label{eq:BVP7}
\end{align}
Here, equations~(\ref{eq:BVP1}) and Eq.~(\ref{eq:BVP2}) are respectively rescaled versions of the (\ref{eqApp:profileGeneral}) and the MSE~(\ref{eqApp:VarProblem}) in time, such that the period $T$ appears explicitly as a parameter in the formulation. In this way, time is rescaled to the interval $[0,1]$ such that the wave profile $y$ has period one; hence, the periodic conditions~(\ref{eq:BVP3}) and (\ref{eq:BVP4}). Condition~(\ref{eq:BVP5}) is an integral phase condition that ensures that translations in solution's time along the periodic solution do not occur during continuation \cite{KrauskopfOsingaBook2007}, where $y_{ref}$ and $\zeta_{ref}$ are the profile and  $\zeta$ from a previous iteration of the continuation scheme.  Since the family of solutions $c \zeta$, $c \in \mathbb{C}$, also satisfy the MSE~(\ref{eqApp:masterstability}), the boundary conditions~(\ref{eq:BVP6}) and (\ref{eq:BVP7}) fixes the solution of the MSE~(\ref{eqApp:masterstability}) to be only one of these; more precisely, the $\zeta$ that satisfies conditions~(\ref{eq:BVP6}) and (\ref{eq:BVP7}) has norm one and $\text{Im}(\pi_1(0))=0$, where $\pi_1$ is a projection onto the first component of $\zeta$ (it fixes the angle of the first component to be zero).  In this way, the set of conditions (\ref{eq:BVP1})--(\ref{eq:BVP7}) generically define a 2PBVP for fixed parameters. 

Numerically, one discretizes $y$ and $\zeta$ in time, such that conditions(\ref{eq:BVP1})--(\ref{eq:BVP7}) define a zero-problem in a suitable finite dimensional space. Then, one can use pseudo-arclength continuation to compute the MSCs presented in Figs.~1--3 in the main manuscript. To achieve this, we make use of the software package \textsc{DDE-Biftool} for Matlab/Octave \cite{Sieber2014},  where the discretization and pseudo-arclength continuation of suitable 2PBVP for delay-advance differential equations is already implemented.

To initialize the continuation scheme, one needs to start from a numerical solution $(y,\zeta,T,\tau,\lambda,\phi):=(y_0,\zeta_0,T_0,\tau_0,\lambda_0,\phi_0)$ that satisfies the 2PBVP. This initial point can be obtained by direct integration or other continuation schemes,  see Subsection~\ref{subSec:Init} for more details. Given the number of conditions in (\ref{eq:BVP1})--(\ref{eq:BVP7}), then four parameters need to be freed up (problem dimension - number of conditions + 1) to compute the stability curves. Since $\lambda \in \mathbb{C}$, then its real and imaginary parts are treated as two separate parameters; hence, one lets $(\phi, \lambda, T)$ be the parameters to be continued. As $(y,\zeta,T,\lambda,\phi)$ varies, we keep track of $\lambda$ until we reach $\lambda_0$ again or $\lambda_0+2i\pi q /T$, for some $q \in \mathbb{Z}$, and then terminate the continuation scheme. After completing the continuation, the family of solutions $(y,\zeta,T,\lambda,\phi)$ is then projected to the $(\mbox{Re}(\lambda),\mbox{imag}(\lambda))$-plane to obtain the MSC. 

We remark that during continuation, we allow $\text{Im}(\lambda)$ to vary beyond the range $[-\pi/T, \pi/T)$ since ---given a solution triple $(\tilde{\lambda},\tilde{\phi},\tilde{\zeta})$ of  (\ref{eqApp:masterstability})--- another solution triple $(\hat{\lambda},\hat{\phi},\hat{\zeta})$ can be obtained via the transformations 
$$(\tilde{\lambda},\tilde{\phi},\tilde{\zeta} \mapsto (\hat{\lambda},\hat{\phi},\hat{\zeta}) = \left(\tilde{\lambda}+\frac{2\pi i q}{T} ,\tilde{\phi}+2\pi q k,\hat{\zeta}\right), \qquad \text{for any } q\in \mathbb{Z},$$
where $\hat{\zeta}(t)=\tilde{\zeta}(t)e^{2\pi i q t/T}$ and $k\in\mathbb{Q}$ is the wave number. Thanks to this transformations, we can always map back triples $(\tilde{\lambda},\tilde{\phi},\tilde{\zeta})$ with $\text{Im}(\tilde{\lambda}) \notin[-\pi/T, \pi/T)$ to triples $(\hat{\lambda},\hat{\phi},\hat{\zeta})$ with $\text{Im}(\hat{\lambda})\in[-\pi/T, \pi/T)$. 

\subsection{Initialization of the two-point boundary value problem}\label{subSec:Init}

In the following, we outline procedures to obtain starting points $(y_0,\zeta_0,T_0,\tau_0,\lambda_0,\phi_0)$ for initialization the continuation subroutine of the 2PBVP~(\ref{eq:BVP1})--(\ref{eq:BVP7}).  To obtain the starting solution for the MSC emanating from the trivial Floquet exponent $\lambda=0$, we employ the following strategy: Given a DTW $\gamma$ obtained by direct simulation, or continuation of a bifurcating periodic solution in system~(\ref{eqApp:ProblemGeneral}), one can readily determine the wave profile $y_0$, its period $T_0$ and wave velocity $1/\tau_0$. Moreover, the solution triple $(\lambda_0,\zeta_0,\phi_0)$ can be chosen as $(\lambda,\zeta,\phi)=(0,y',0)$ where $y'$ can be computed using numerical differentiation and normalized. The thus constructed solution $(y_0,\zeta_0,T_0,\tau_0,\lambda_0,\phi_0)$ effectively corresponds to the solution generated by solving the 2PBVP~(\ref{eq:BVP1})--(\ref{eq:BVP7}) for $\lambda=0$, and can be used to start the continuation routine.

To compute the initial solutions for other MSCs, one can always make use of the monodromy matrix $M=\Xi(T)$ (recall that $\Xi(t)$ is the fundamental matrix solution of (\ref{eqApp:VarProblem})) to compute tuples $(\mu,z)$ that could be used to construct the starting point by using Theorem~\ref{thm:MainSupp}. This approach is suitable for small networks (e.g. in a minimal network realization of a given wave profile) and one is interested in the stability of higher-order embeddings of the profile. Furthermore, these starting solutions can be extended to larger networks with different wavenumber $k$ by continuation, hence saving the cost of recomputing the monodromy matrix, see Sec.~\ref{sec:OtherApp}.

For large networks, however, it is generally more convenient to use the fact that Floquet exponent and bundle pairs $(\mu,z)$, with $\mu=e^{\lambda T}$, can be associated with a given perturbation wave number $\phi/(2\pi).$ Specifically, Theorem 1(iii) implies that generically the bundle $z$ has spatio-temporal symmetry $z_{n+1}(t+\tau)=e^{\lambda \tau + i\phi}z_n(t)$, for all $n=1,...,N-1$ and $z_1(t+\tau)=e^{\lambda \tau +i\phi} z_N(t)$. Thus, $z$ is 
also a solution to the twisted boundary value problem 
\begin{equation}\label{eqApp:VarProblem-twisted}
\begin{aligned}
    z^\prime_n(t) &= \sum_{0 \leq |m| \leq r} A_m(t-n\tau)z_{n+m}(t)\\
    z_n(t) &=e^{-\lambda \tau - i\phi}z_{n+1}(t+\tau),
\end{aligned} \tag{S17} 
\end{equation}
where indices are again understood as modulo $N.$ Without the boundary conditions, Eq.~(\ref{eqApp:VarProblem}) and Eq.~(\ref{eqApp:VarProblem-twisted}) defined the same initial value problem (they are identical), thus they share the same fundamental matrix solution $\Xi$. The (twisted) monodromy matrix $\Xi(\tau)$ \cite{deWolff2023equivariant}, associated to Eq.~(\ref{eqApp:VarProblem-twisted}), can be treated with more efficiently than $M$ by choosing a suitable basis representation. Specifically, by restricting the action of $\Xi(\tau)$ to the eigendirections corresponding to a specific perturbation with associated wave number $\phi/2\pi,$ we can compute $\lambda$ as an eigenvalue of a $d\times d$-dimensional matrix (the reduced monodromy matrix) $$R(\phi):=V(\phi)^H\Xi(\tau)V(\phi)\in \mathbb{C}^{d\times d},$$  where $$V(\phi):=\left((e^{i\phi},e^{i2\phi},\ldots,1)\otimes \mbox{Id}_d\right)/\sqrt{N} \in \mathbb{C}^{Nd\times d},$$ the symbol $\otimes$ denotes the Kronecker product, and $V(\phi)^H$ is the Hermitian (or conjugate) transpose of $V(\phi)$. The following proposition makes this statement rigorous.

\begin{preposition}[Floquets exponents corresponding to given wave number]\label{prep:PlaneWave} 
For a network of size $N$, let $\gamma$ be a DTW solution of system~(\ref{eqApp:ProblemGeneral})  with wave profile $y$, wave speed $1/\tau$, wave number $k$, and period $T$ (as in Theorem 1). The following statements hold:
\begin{enumerate}[(i)]
\item Let $\mu^*=e^{\lambda^\ast T}$ (where $\lambda^\ast$ is defined as $\lambda = \text{Log} (\mu)/T$) be a simple Floquet multiplier of $\gamma$ with associated bundle $z^*$ and perturbation wave number $\phi^\ast/(2\pi)$ ($\phi^\ast$ defined as in Theorem 1(iii)). Then $(\mu^\ast)^k=e^{\lambda^\ast \tau}$ is a eigenvalue of the reduced monodromy matrix $R(\phi^\ast)$ with associated eigenvector $\alpha=V(\phi^\ast)^Hz^\ast(0).$

\item Let $\phi=2\pi l/N$, for some $l=1,...N$, and let $b=e^{\lambda \tau}$ (where $\lambda$ is defined as $\lambda = \text{Log} (b)/\tau$) be an eigenvalue with associated eigenvector $\beta$ of the reduced monodromy matrix
$R(\phi).$ Then $\mu=e^{\lambda T}$ is a Floquet multiplier of $\gamma$ with associated bundle $z(t)=\Xi(t)V(\phi)\beta$ and perturbation wave number $\phi/(2\pi)$.
\end{enumerate}
\end{preposition}

\begin{proof} We begin with preliminary remarks. Note that, whenever $\phi=2\pi l/N$ for some $l=1,...,N$,  the vectors $$\nu_j(\phi):=\frac{1}{\sqrt{N}}\left((e^{i\phi},e^{i2\phi},\ldots,1)\otimes e_j\right),\quad (e_j)_l=\delta_{j,l},\quad e_j\in\mathbb{R}^d,  \quad \nu_j(\phi^*)\in\mathbb{C}^{Nd},$$ are eigenvectors of the matrix $C\otimes \mbox{Id}_d$ with corresponding eigenvalue $e^{-i\phi}$, where \begin{equation*}
C=\left(\begin{matrix}
    0 & 0 & \dots  & 0 & 1\\
    1 & 0 & \ldots & 0 & 0\\
    0 & 1 & \ldots  & 0 & 0\\
    \vdots & \vdots & \ddots  & \vdots & \vdots \\
    0 & 0 & \ldots  & 0 & 0\\
    0 &  0 & \ldots  & 1 & 0
\end{matrix}\right)
\end{equation*} is the $N$-dimensional circulant (companion) matrix defined earlier in the text and $\otimes$ is the Kronecker product, since
$$(C\otimes \mbox{Id}_d)\nu_j(\phi)=(C\otimes \mbox{Id}_d)(e^{i\phi},e^{i2\phi},\ldots,1)\otimes e_j)=(C(e^{i\phi},e^{i2\phi},\ldots,1))\otimes (\mbox{Id}_de_j)=e^{-i\phi}(e^{i\phi},e^{i2\phi},\ldots,1)\otimes e_j=e^{-i\phi}v_j(\phi).$$
Notice also, that $V(\phi)=[\nu_1(\phi)|\nu_2(\phi)|...|\nu_d(\phi)]$, that is, the columns of $V(\phi)$ correspond to the vectors $\nu_j(\phi)$. 
Since $V(\phi)$ consists of linearly independent column vectors, there exists a left (Moore-Penrose) pseudo-inverse $V(\phi)^+=(V(\phi)^H V(\phi))^{-1}V(\phi)^H,$ where $V(\phi)^H$ is the hermitian (or conjugate) transpose. Specifically, $V(\phi)^+=V(\phi)^H,$ since 
\begin{align*}
    V(\phi)^H V(\phi)&=\frac{1}{N}\left(((e^{i\phi},e^{i2\phi},\ldots,1)\otimes \mbox{Id}_d)^H ((e^{i\phi},e^{i2\phi},\ldots,1)\otimes \mbox{Id}_d)\right)\\
    &=\frac{1}{N}\left(((e^{i\phi},e^{i2\phi},\ldots,1)^H\otimes \mbox{Id}_d) ((e^{i\phi},e^{i2\phi},\ldots,1)\otimes \mbox{Id}_d)\right)\\
    &=\frac{1}{N}\left(((e^{i\phi},e^{i2\phi},\ldots,1)^H(e^{i\phi},e^{i2\phi},\ldots,1))\otimes (\mbox{Id}_d\mbox{Id}_d)\right)=\mbox{Id}_d.
\end{align*}
Analogously, it can be shown that $V(\phi)^H$ is also a right inverse of $V(\phi)$, i.e. $V(\phi)V(\phi)^H=\text{Id}_{Nd}.$ We continue with the proof of the proposition.

\begin{itemize}
\item[$(i)\Rightarrow(ii)$] Let $(\mu^\ast,z^*)$ be a solution tuple for the Floquet problem~(\ref{eqApp:VarProblem}) with $\mu^*$ simple and associated Floquet exponent $\lambda^\ast=\mbox{Log}(\mu^\ast)/T$. By definition, $z^\ast(0)$ is eigenvector of the monodromy matrix M corresponding to the eigenvalue $\mu^\ast$, i.e. $z^\ast(T)=Mz^\ast(0)=\mu^\ast z^\ast(0).$ 

\underline{Step 1}: We show that $z^\ast(0)$ is also an eigenvector to the (twisted) monodromy matrix $\Xi(\tau)$ with associated eigenvalue $(\mu^\ast)^{k}=e^{\lambda^\ast\tau}$, where $k$ is the wave number of DTW $\gamma$. Recall that $k=M_0/N_0$, where $M_0,N_0\in\mathbb{N}$ with $\gcd(M_0,N_0)=1$, and $\tau=kT$. It holds that  $\Xi(\tau)^{N_0}=\Xi(T)^{M_0}=M^{M_0}$ and, as a result, eigenvectors of the twisted monodromy matrix $\Xi(\tau)$ correspond to eigenvectors of the monodromy matrix $M$ and vice versa. Moreover, $\Xi(\tau)^{N_0}z^\ast(0)=M^{M_0}z^\ast(0)=\mu^{M_0}z^\ast(0)$ implying that (by taking the fundamental $N_0$-th root) $$\Xi(\tau)z^\ast(0)=e^{\lambda^\ast TM_0/N_0}z^\ast(0)=e^{\lambda^\ast kT}z^\ast(0)=e^{\lambda^\ast \tau}z^\ast(0).$$

\underline{Step 2}: We show that $z^\ast(0)$ is a linear combination of vectors $\nu_j(\phi^*),~j=1,\ldots,d$. Theorem 1(iii) implies there there exist $\phi^\ast=2\pi l^*/N$, for some $l^*=1,...,N$, (and $\zeta^\ast$) such that 
$$z_{n+1}^\ast(t+\tau)=\zeta^\ast(t-(n+1)\tau+\tau)e^{\lambda^* (t+\tau)+ i(n+1)\phi^*}=e^{\lambda^* \tau +i\phi^*}\zeta^\ast(t-n\tau)e^{\lambda^* t+ in\phi}=e^{\lambda^* \tau + i\phi^\ast}z^\ast_n(t)$$ 
for all $n=1,...,N-1$ and $z^\ast_1(t+\tau)=e^{\lambda^* \tau +i\phi^\ast} z^\ast_N(t).$ As a result, the bundle $z^\ast$ satisfies the twisted boundary condition $z^\ast(\tau)=e^{\lambda^*\tau + i\phi^\ast}[C\otimes \mbox{Id}_d]z^\ast(0)$. Therefore, and by using that $e^{\lambda^*\tau}$ is an eigenvalue of $\Xi(\tau)$ with eigenvector $z^\ast(0)$, we have
$$[C\otimes \mbox{Id}_d]z^\ast(0)=e^{ -i\phi^\ast -\lambda^\ast\tau}e^{\lambda^*\tau + i\phi^\ast}[C\otimes \mbox{Id}_d]z^\ast(0)=e^{ - i\phi^\ast -\lambda^\ast\tau}z^\ast(\tau)=e^{-i\phi^\ast}z^\ast(0)$$
implying that $z^\ast(0)$ can be represented as an eigenvector of the matrix $[C\otimes \mbox{Id}_d]$ corresponding to the eigenvalue $e^{-i\phi^\ast}$. Given that the matrix $[C\otimes \mbox{Id}_d]$ has exactly $d$-many linearly independent eigenvectors corresponding to the eigenvalues $e^{-i\phi^\ast}$, which are of the form $\nu_j(\phi^\ast),$ it follows that $z^\ast(0)$ must be representable as a linear combination of these vectors, i.e. $z^\ast(0)=\sum_{j=1}^d \alpha_j\nu_j(\phi^\ast)$, with unknown coefficients $\alpha_j,~j=1,\ldots,d$.  

\underline{Step 3}: Derivation of the reduced matrix problem. Let $V(\phi^\ast)$ be as in (ii) and $\alpha$ denote the column vector of coefficients $\alpha_j$, $j=1,\ldots,d$, such that $z^\ast(0)=V(\phi^\ast)\alpha.$ Therefore, $$z^\ast(\tau)=\Xi(\tau)z^\ast(0)=\Xi[\tau]V(\phi^\ast)\alpha.$$

By again using that $e^{\lambda^\ast \tau}$ is an eigenvalue of $\Xi(\tau)$ corresponding to $z^\ast(0)$, i.e. $z^\ast(\tau)=e^{\lambda^\ast \tau} z^\ast(0),$ it follows that $$\Xi[\tau]V(\phi^\ast)\alpha=e^{\lambda^\ast \tau}V(\phi^\ast)\alpha.$$
As a result, we have that $$V(\phi^\ast)^H\Xi[\tau]V(\phi^\ast)\alpha = e^{\lambda^\ast \tau}\alpha,$$ and thus, $\alpha=V(\phi^\ast)^Hz^\ast(0)$ is an eigenvector of $R(\phi^\ast)=V(\phi^\ast)^H\Xi[\tau]V(\phi^\ast)$ with associated eigenvalue $e^{\lambda^\ast \tau}.$

\item[$(ii)\Rightarrow(i)$] Assume as in (ii). Since $\beta$ is an eigenvector of $R(\phi)$ with eigenvalue $b$, then
\begin{equation*}
    \begin{aligned}V(\phi)^H\Xi[\tau]V(\phi)\beta &= b\beta, \\
    \Xi[\tau]V(\phi)\beta &= bV(\phi)\beta,
    \end{aligned}
\end{equation*}
where we used the identity $V(\phi)V(\phi)^H=\text{Id}_{Nd}$. Therefore, $V(\phi)\beta$ is an eigenvector of the twisted monodromy matrix $\Xi(\tau)$ associated with the eigenvalue $b.$ As a result, $b$ is also an eigenvector of the monodromy matrix $M$ with corresponding eigenvalue $b^{1/k}$ (by the same reasoning as employed in (i) Step 1).
The bundle $z^\ast$ is computed as $z^\ast(t)=Z(t,b)=\Xi(t)V(\phi)\beta$. Then the assertion follows.

\end{itemize}
\end{proof}

Proposition 1 motivates the following strategy for effectively computing Floquet exponents of DTW  $\gamma$ in a network of size $N$ corresponding to a fixed $\phi=2\pi l/N$ for some $l=1,\ldots,N$:
\begin{enumerate}
    \item Numerically solve the set of initial value problems $z(0)=\nu_j(\phi), $ $j=1,...,d,$ where $$\nu_j(\phi^\ast):=\frac{1}{\sqrt{N}}\left((e^{i\phi^\ast},e^{i2\phi^\ast},\ldots,1)\otimes e_j\right),\quad (e_j)_l=\delta_{j,l},\quad e_j\in\mathbb{R}^d,  \quad \nu_j(\phi^*)\in\mathbb{C}^{Nd}$$ associated with the original Floquet problem (\ref{eqApp:VarProblem}) until time $\tau$ to obtain vectors $Z(\tau,\nu_j(\phi))=\Xi(\tau)\nu_j(\phi)$.
    \item Collect the (column) vectors $Z(\tau,\nu_j(\phi)),~j=1,\ldots,d$ in a $Nd\times d$-dimensional matrix $$\Xi[\tau]V(\phi)=(Z(\tau,\nu_1(\phi))|\ldots| Z(\tau,\nu_d(\phi))).$$ 
    \item Compute the eigenvalue-eigenvector pairs $(b,\beta)$ of the matrix $V(\phi)^H\Xi[\tau]V(\phi)$ with associated exponents $\lambda=\mbox{Log}(b)/\tau,$ i.e. $(b,\beta)$ such that $V(\phi)^H\Xi[\tau]V(\phi)\beta=b\beta$. 
    \item Compute the bundle $z$ corresponding to $\mu=e^{\lambda T}$ as the solution $z(t)=Z(t,V(\phi)\beta)=\Xi(t)V(\phi)\beta$ corresponding to the initial value problem $z(0)=V(\phi)\beta$.
\end{enumerate}

\subsection{Other applications}\label{sec:OtherApp}
Thus far, we have presented a numerical framework for the computation of MSCs, which enables studying the stability of families of DTWs sharing the same profile $y$, period $T$ and wave number $k$, across different network sizes. The focus of this treatment was the determination of the Floquet exponent spectra without the need to compute the spectrum of each DTW every time the network size changes.

In the following, we describe other situations in which continuation of solutions to the 2PBVP~(\ref{eq:BVP1})–(\ref{eq:BVP7}) can be employed by utilizing different continuation parameters or incorporating additional boundary conditions, depending on the specific situation. For example, in the main manuscript, we demonstrated how the continuation of the 2PBVP~(\ref{eq:BVP1})–(\ref{eq:BVP7})  allows studying the relative position and curvature of the relevant MSC with respect to the imaginary axis evolves as the wave number $k$ changes. This is illustrated in Figs.~2 and 3 of the main manuscript. In this manner, we are able to identify destabilizations of the DTW as a function of the wave number, pinpointing the families of different network sizes where the DTW becomes unstable, as well as the existence of bound two-pulse DTWs in those networks.

Furthermore, going beyond the results in the main manuscript, the  2PBVP formulation~(\ref{eq:BVP1})–(\ref{eq:BVP7}) can also be used in the following cases:
\begin{itemize}
\item Continuing in one (or more) system parameters by individually fixing any (or more) of the solution parameters $T$, $\text{Re}(\lambda)$, $\text{Im}(\lambda)$, or $\phi$. This way one is able to continue in system parameter, e.g. homoclinic bifurcations (by fixing a sufficiently large period $T$), bifurcations corresponding to the crossing of Floquet exponents of the unit circle (by fixing $\text{Re}(\lambda)$) and the destabilization of the DTW with respect to perturbations of a specific wave number $\phi$.
\item  Improving performance when continuing bifurcations in system parameters by replacing the extended $2Nd$-dimensional two-point boundary problem (consisting of (\ref{eqApp:ProblemGeneral}) and (\ref{eqApp:VarProblem})) with the $2d$-dimensional 2PBVP. What is more, for large system sizes, we typically expect low-dimensional bifurcations of discrete traveling waves (e.g. torus bifurcations) to not be isolated but come in families in analogy to their continuous counterparts (such as Turing bifurcations). Indeed, our method provides a unique angle on computing these spatially extended instabilities as the crossing of an MSC of the imaginary axis (together with the associated FE spectrum on it).
\item Approximating the changes of curvature of the MSC close to $\lambda=0$, to efficiently compute when the DTW is unstable.  This can be done,e.g., by fixing $\phi$ to a small number and detecting sign changes of $\text{Re}(\lambda)$ as continuation parameters are varied.
\item Calculating effectively the full Floquet exponent spectrum, as well as the relevant parts of it, in specific applications. Firstly, Proposition 1 allows for computing FEs corresponding to a specific wave number to understand whether perturbations with this (spatial) wave number will grow or decay. This is independent of the remaining part of the spectrum, which can be stable or unstable and typically requires more sophisticated methods (such as Gram-Schmidt orthogonalization), which again is computationally expensive for large systems.
The MSC can then be used to also investigate effectively the stability of nearby wave numbers and, therefore, to determine different types of instabilities that must have occurred together with the typical phenomena associated with them. In situations, where the full spectrum is required (such as e.g. computing isostable curves \cite{nicks2024insights}) the master stability framework can be used to effectively compute the spectrum independently of the network size.
\end{itemize}


\end{document}